\documentclass[copyright,creativecommons]{eptcs}
\providecommand{\event}{QPL 2022} % Name of the event you are submitting to

\usepackage{iftex}

\ifpdf
  \usepackage{underscore}         % Only needed if you use pdflatex.
  \usepackage[T1]{fontenc}        % Recommended with pdflatex
\else
  \usepackage{breakurl}           % Not needed if you use pdflatex only.
\fi

\usepackage{amsthm, amsmath, amssymb, dsfont, physics, textcomp, tikz}
\usepackage{braket}
\usepackage{qcircuit}
\usepackage{tikz-cd}
\usepackage{hyperref}
\usepackage{stmaryrd}
\usepackage[capitalize,nameinlink,sort,noabbrev]{cleveref}
\usepackage{algorithm}
\usepackage{algpseudocode}
\usepackage{booktabs}

\newcommand{\Z}{\mathbb{Z}}

\newcommand{\C}{\mathbb{C}}

\newcommand{\R}{\mathbb{R}}

\newcommand{\s}[1]{\{#1\}}

\newcommand{\comp}[1]{\textbf{#1}}

\newcommand{\cxgate}{\textsc{cnot}}
\newcommand{\ccxgate}{\textsc{tof}}
\newcommand{\swapgate}{\textsc{swap}}
\newcommand{\czgate}{\textsc{cz}}
\newcommand{\pgate}{\textsc{s}}
\newcommand{\hgate}{\textsc{h}}
\newcommand{\tgate}{\textsc{t}}
\newcommand{\xgate}{\textsc{x}}
\newcommand{\igate}{\textsc{i}}

\newtheoremstyle{break}% name
  {}%         Space above, empty = `usual value'
  {}%         Space below
  {\itshape}% Body font
  {}%         Indent amount (empty = no indent, \parindent = para indent)
  {\bfseries}% Thm head font
  {.}%        Punctuation after thm head
  {\newline}% Space after thm head: \newline = linebreak
  {}%         Thm head spec

\theoremstyle{plain}
\newtheorem{theorem}{Theorem}[section]
\newtheorem{lemma}[theorem]{Lemma}
\newtheorem{proposition}[theorem]{Proposition}
\newtheorem*{proposition*}{Proposition}
\newtheorem{corollary}[theorem]{Corollary}
\theoremstyle{break}

\theoremstyle{definition}
\newtheorem{definition}[theorem]{Definition}

\newtheorem{example}[theorem]{Example}
\theoremstyle{remark}
\newtheorem{remark}[theorem]{Remark}

\usepackage{xcolor}
\hypersetup{
    colorlinks,
    linkcolor={red!80!black},
    citecolor={green!80!black},
    urlcolor={blue!80!black}
}

\begin{document}

% --------------------------------------------------------------------
\title{Symbolic Synthesis of Clifford Circuits and Beyond}

\author{Matthew Amy
\institute{School of Computing Science\\ Simon Fraser University, Burnaby, Canada}  
\email{meamy@sfu.ca}
\and
Owen Bennett-Gibbs
\institute{Department of Mathematics and Statistics\\ McGill University, Montreal, Canada}
\email{owen.bennett-gibbs@mail.mcgill.ca}
\and Neil J. Ross
\institute{Department of Mathematics and Statistics \\ Dalhousie University, Halifax, Canada}
\email{neil.jr.ross@dal.ca}
}
\newcommand{\titlerunning}{Symbolic Synthesis of Clifford Circuits and Beyond}
\newcommand{\authorrunning}{M. Amy, O. Bennett-Gibbs \& N.J. Ross}

% --------------------------------------------------------------------
\maketitle

% ------------------------------------------------------------------
\begin{abstract}
Path sums are a convenient symbolic formalism for quantum operations with applications to the simulation, optimization, and verification of quantum protocols. Unlike quantum circuits, path sums are not limited to unitary operations, but can express arbitrary linear ones. Two problems, therefore, naturally arise in the study of path sums: the unitarity problem and the extraction problem. The former is the problem of deciding whether a given path sum represents a unitary operator. The latter is the problem of constructing a quantum circuit, given a path sum promised to represent a unitary operator. 

In this paper, we show that the unitarity problem is \textbf{co-NP}-hard in general, but that it is in \textbf{P} when restricted to Clifford path sums. We then provide an algorithm to synthesize a Clifford circuit from a unitary Clifford path sum. The circuits produced by our extraction algorithm are of the form $C_1HC_2$, where $C_1$ and $C_2$ are Hadamard-free circuits and $H$ is a layer of Hadamard gates. We also provide a heuristic generalization of our extraction algorithm to arbitrary path sums. While this algorithm is not guaranteed to succeed, it often succeeds and typically produces natural looking circuits. Alongside applications to the optimization and decompilation of quantum circuits, we demonstrate the capability of our algorithm by synthesizing the standard quantum Fourier transform directly from a path sum.
\end{abstract}

% ------------------------------------------------------------------
\section{Introduction}
\label{sec:introduction}

The circuit model is ubiquitous in quantum computing, from hardware assembly code to the high-level description of algorithms. Quantum compilation typically amounts to a series of circuit-to-circuit transformations, lowering a circuit, described programmatically in a \emph{circuit description language} over a high-level gate set, down through a series of progressively restrictive gate sets with more fault-tolerance and hardware constraints. Accordingly, quantum algorithms are frequently described at the level of quantum circuits, plugging together inputs and outputs of large, complicated circuits. An exception is the classical oracles used in many quantum algorithms, which are often described at the level of classical programs or Boolean logic, and then \emph{synthesized} as a high-level quantum circuit.

Despite this, the circuit model is often a less than ideal representation of quantum computations. Semantically, circuits expose little information about a computation to the naked eye, particularly when low-level gate sets like Clifford+$T$ are used. Likewise, reasoning about quantum circuits is often a difficult affair involving re-write rules derived from circuit \emph{relations}. Complete sets of relations are only known for a small number of low-level non-universal gate sets \cite{s15, acr17,mrs21}, or higher-level gate sets \cite{bs21,lrs21} which result in a large degree of overhead when compiled. Even with complete sets of relations, simplification of quantum circuits using re-write rules is costly and highly local, yielding results which are typically far from optimal.

Recently, alternative models of quantum computation such as those based on diagrammatic calculi \cite{cd08,bk18}, have risen in popularity. These models have seen success in circuit simplification, among other applications, due in part to more effective re-writing methods. However, as most existing quantum computers ultimately run on circuit-like languages, a key component in using such models for circuit transformations is the ability to synthesize or \emph{extract} a circuit back from the representation. This problem has seen a great deal of attention recently in the context of graphical calculi, resulting in methods for the extraction of Clifford and Clifford+$T$ circuits from ZX diagrams admitting a \emph{generalized flow} \cite{dkpw20,bmflw21}, as well as theoretical results studying the hardness of this extraction problem \cite{bkw22}.

In this paper, we study the problem of synthesizing a unitary circuit given a symbolic expression of a linear operator as a \emph{sum-over-paths} or \emph{path sum} \cite{a18}. As any linear operator between $2^n$-dimensional Hilbert spaces is representable in this form, we first consider the problem of deciding whether a path sum representations a unitary transformation and show that it is generically \textbf{coNP}-hard. Restricting to path sums representing Clifford operators we show that the unitarity problem is in \textbf{P}, and that a unitary circuit can be synthesized in time polynomial in the number of qubits. This extraction algorithm produces circuits of the form $C_1HC_2$, where $C_1$ and $C_2$ are Hadamard-free circuits decomposable as a product of $\pgate$, $\xgate$, $\czgate$, and $\cxgate$ gates, and $H$ is a layer of Hadamard gates. As a consequence we obtain a simple, constructive proof of the $7$-stage Bruhat decomposition of the Clifford group \cite{bm21,mr18}.

For non-Clifford operations, we give a heuristic for the unitary synthesis of general sums. While our heuristic does not always produce a circuit even if the path sum represents a unitary operator, it succeeds often in practice and typically returns efficient, natural circuits. Alongside circuit optimization, this heuristic has applications to the \emph{decompilation} of quantum circuits, whereby a circuit over a low-level gate set such as Clifford+$T$ is re-written over a higher-level gate set such as multiply-controlled Toffoli gates. We further demonstrate the capability of our algorithm by synthesizing the typical quantum Fourier transform circuit directly from its specification as a sum-over-paths.

% ------------------------------------------------------------------
\section{Path sums}
\label{sec:pathsum}

We begin by briefly reviewing the theory of \emph{path sums} \cite{a18,v21}. A \emph{path sum representation} of a linear operator $\Psi:\C^{2^m}\rightarrow \C^{2^n}$ is an expression for $\Psi$ as a sum indexed by binary variables such as
\begin{equation}
  \label{eq:pathsum1}
	\Psi\ket{\vec{x}} = \mathcal{N}\sum_{\vec{y}\in\Z_2^{k}}e^{2\pi i P(\vec{x}, \vec{y})}\ket{f(\vec{x}, \vec{y})},
\end{equation}
where $\mathcal{N}\in\C\setminus \s{0}$ is a \emph{normalization factor}, and $P:\Z_2^m\times \Z_2^k \rightarrow \R$ and $f:\Z_2^m\times \Z_2^k \rightarrow \Z_2^n$ are real- and Boolean-valued multilinear polynomials, respectively. The path sum in \cref{eq:pathsum1} is said to be \emph{amplitude-balanced} because the normalization factor $\mathcal{N}$ is independent of $\vec{x}$ and $\vec{y}$. We sometimes denote the path sum representation of an operator $\Psi$ by $\ket{\Psi}$.

\cref{eq:pathsum1} provides a representation the operator $\Psi$ in the sense that instantiating the binary variables $\vec{x}$ and $\vec{y}$ on both sides of the equality yields a true equation. For this reason, we think of a path sum as a symbolic description of the action of a linear operator on computational basis states. 

\begin{example}
  \label{ex:cliffordt1}
  The phase gates $\pgate$ and $\tgate$, as well as the Hadamard gate $\hgate$, can be represented by path sums as follows, where $\omega=e^{i\pi/4}$:
  \begin{itemize}
  \item $\pgate\ket{x} = i^x\ket{x}$,
  \item $\tgate\ket{x} = \omega^x\ket{x}$, and
  \item $\hgate\ket{x} = \frac{1}{\sqrt{2}}\sum_{y}(-1)^{xy}\ket{y}$.
  \end{itemize}
\end{example}

As path sums involve arithmetic and polynomials over Boolean variables in various rings, it is useful to recall that Boolean algebra can be embedded in any (unital) ring $\mathcal{R}$ using the \emph{lifting} construction defined in \cite[Lemma 7.1.6]{a18} and reproduced below.
\begin{align*}
	\hspace*{5em} \overline{0} &= 0_\mathcal{R} \hspace{-5em} & \overline{f \land g} &= \overline{f} \cdot \overline{g} \\
	\hspace*{2em} \overline{1} &= 1_\mathcal{R}  \hspace{-5em} & \overline{f \oplus g} &= \overline{f} + \overline{g} - (2 \cdot \overline{f} \cdot \overline{g})
\end{align*}
Lifting allows one to use Boolean expressions of variables inside path sums coherently, leading to more natural expressions, as in the following example.

\begin{example}
  \label{ex:cliffordt2}
  The gates $\cxgate$ and $\ccxgate$ admit the following path sum representations:
  \begin{itemize}
  \item $\cxgate\ket{x_1x_2} = \ket{x_1}\ket{x_1\oplus x_2}$ and
  \item $\ccxgate\ket{x_1x_2x_3} = \ket{x_1}\ket{x_2}\ket{x_1\oplus (x_2\cdot x_3)}$.   
  \end{itemize}
\end{example}

The sum-over-paths representations of linear operators has been studied extensively in the context of quantum information \cite{dhmhno05,bvr08,r09,m17,kps17,bh21}. Recent work on the connection to graphical calculi \cite{lwk20,v21} has shown that path sums form a universal model for linear operators over $2^n$-dimensional Hilbert spaces through direct translations from universal graphical calculi such as the ZH-calculus \cite{bk18}.

\begin{proposition}[Unversality]
  \label{prop:universality}
  Any linear operator $\Psi:\C^{2^m}\rightarrow \C^{2^n}$ admits a representation as a path sum.
\end{proposition}

\begin{example}
  \label{ex:compact}
  Path sums can represent linear operators between spaces of different
  dimensions. The operators $\eta:\C \rightarrow \C^{2}\otimes \C^2$
  and $\varepsilon:\C^{2}\otimes \C^2 \rightarrow \C$, which act,
  respectively, as the unit and counit in the category \textbf{FdHilb}
  \cite{v21}, can be written as the following path sums:
  \begin{itemize}
  \item $\eta\ket{} = \sum_y\ket{yy}$ and
  \item $\varepsilon\ket{x_1x_2} = \frac{1}{2}\sum_y(-1)^{y(x_1 + x_2)}\ket{}$.   
  \end{itemize}
  When a path sum represents a row or column vector as above, we drop
  any empty $\ket{}$. A path sum with a single output dimension, 
  representing a $\C$-valued linear map, is said to be
  \emph{dimensionless}.
\end{example}

In contrast to graphical calculi, which have a \emph{compositional} structure, path sums are effectively \emph{global} expressions of a linear operator. In other words, the composition (sequential or parallel) of two linear operators is reified into an expression of the form of \cref{eq:pathsum1}. This is accomplished through the substitution of \emph{free variables} --- variables that are not summed over, corresponding to inputs of the operator as elements of the computational basis. We denote the free variables of a path sum $\ket{\Psi}$ by $FV(\ket{\Psi})$. A path sum with free variables may be thought of as a symbolic state vector in indeterminates $\vec{x}=FV(\ket{\Psi})$. Hence we use the notation $\ket{\Psi(\vec{x})}$ to denote a path sum expression for the operator $\Psi$ with free variables $\vec{x}$. We use $\ket{\Psi(x)}$ to denote a path sum with a distinguished free variable $x$.

Through the lifting operation described above, we can define a notion of substitution for path sums with free variables. In particular, given a free variable $x$ appearing in a path sum we may substitute $x$ with any Boolean expression $f$ in all relevant contexts (the phase or the state).

\begin{definition}[Substitution]
	Let $\ket{\Psi(x)}$ be a path sum with free variable $x$ and let $f$ be a Boolean expression. Then the \emph{substitution} of $x$ with $f$ is denoted $\ket{\Psi(f)}$.
\end{definition}

Reasoning with local binders and free variables in the path sums requires care to avoid variable capture. For instance, let $\ket{\Psi(x)} = \sum_y\ket{x}$. It can be observed that $\ket{\Psi(x)}$ represents the linear operator $\Psi = 2I$. However, if $x$ is substituted with the free variable $y$, the $\sum_y$ captures $y$, giving the path sum $\ket{\Psi(y)} = \sum_y\ket{y}$, representing the vector $\ket{0} + \ket{1}$. We assume that substitution is capture-avoiding unless otherwise noted.

A variable may also be \emph{bound} by summing over its possible values. A bound variable may be locally viewed as a free variable by pulling the summation outside of an expression. Indeed, if $\ket{\Psi(x)}$ is a path sum with free variable $x$, then $\sum_x\ket{\Psi(x)}$ is a path sum with free variables $FV(\ket{\Psi})\setminus\{x\}$. We sometimes refer to bound variables as \emph{path} variables.

\begin{example}
  \label{ex:had2}
  Recall that the Hadamard gate can be represented as $\hgate\ket{x} = \frac{1}{\sqrt{2}}\sum_{y}(-1)^{xy}\ket{y}$. Alternatively, the Hadamard gate can also be written as $\hgate\ket{x} = \sum_y \ket{\Psi(x, y)}$
  where  $\ket{\Psi(x, y)} = \frac{1}{\sqrt{2}}(-1)^{xy}\ket{y}$ is a path sum in the free variables $x$ and $y$.
\end{example}

By \cref{prop:universality}, any linear operator admits a path sum representation. In particular, the composition or tensor product of any two linear operators can also be represented as a path sum. The following proposition gives explicit expression for these constructions in the language of path sums.

\begin{proposition}[Parallel \& Sequential composition]
  Let $\Psi:\C^{2^m} \rightarrow \C^{2^n}$ and $\Phi:\C^{2^{s}} \rightarrow \C^{2^t}$ be two linear operators and let $\Psi\ket{\vec{x}} = \mathcal{N}\sum_{\vec{y}\in\Z_2^{k}}e^{2\pi i P(\vec{x}, \vec{y})}\ket{f(\vec{x}, \vec{y})}$ and $\Phi\ket{\vec{w}} = \mathcal{M}\sum_{\vec{z}\in\Z_2^{l}}e^{2\pi i Q(\vec{w}, \vec{z})}\ket{g(\vec{w}, \vec{z})}$ be expressions of $\Psi$ and $\Phi$ as path sums. Then
\begin{align*}
  	\Psi\otimes\Phi\ket{\vec{x}}\ket{\vec{w}} &= 
 	 \mathcal{NM}\sum_{\vec{y}\in\Z_2^{k}}\sum_{\vec{z}\in\Z_2^{l}}e^{2\pi i \left[P(\vec{x}, \vec{y}) + Q(\vec{w}, \vec{z})\right]}\ket{f(\vec{x}, \vec{y})}\otimes \ket{g(\vec{w}, \vec{z})} \\
	\Phi\circ\Psi\ket{\vec{x}} &= 
  \mathcal{N}\sum_{\vec{y}\in\Z_2^{k}}e^{2\pi i P(\vec{x}, \vec{y})}\ket{\Phi(f(\vec{x}, \vec{y}))}
\end{align*}
	as path sums, where the latter is well-formed if and only if $s=n$.
\end{proposition}

The parallel and sequential composition of path sums provides a method to compute a symbolic expression for a circuit over a set of basic gates with known path sums. Moreover, for typical gate sets of interest, this representation has size polynomial in the size of the circuit. We give one such result below \cite[Corollary 2.15]{a18} for the class of circuits which will be most relevant for the purposes of this paper. 

\begin{proposition}[Efficiency for Clifford+$T$]
Any circuit over Clifford+$T$ gates of volume $V$ can be expressed as a path sum which has size polynomial in $V$ and can be computed in time polynomial in $V$.
\end{proposition}

\paragraph{Equational reasoning}
A major utility of the path sum representation \cite{a18} comes from the ability to perform \emph{equational reasoning}. Complete equational theories of Clifford unitaries \cite{a18} and more general stabilizer operations \cite{v21} have previously been developed. We reformulate these theories here using locally free variables to simplify their presentation.

\begin{proposition}\label{prop:eqns}
Let $\Psi$ be a path sum such that $y\notin FV(\Psi)$ and let $f$ be a Boolean expression such that $x,y\notin FV(f)$. Then the following equations hold.
\begin{align}
	\sum_y\ket{\Psi} &= 2\ket{\Psi} \label{eq:e} \\
	\sum_{x,y}(-1)^{y(x + f)}\ket{\Psi(x)} &= 2\ket{\Psi(f)} \label{eq:i} \\
	\sum_yi^y(-1)^{yf}\ket{\Psi} &= \omega\sqrt{2}(-i)^{f}\ket{\Psi} \label{eq:u} \\
	\sum_y\ket{\Psi(y)} &= \sum_y\ket{\Psi(y + f)} \label{eq:v}
\end{align}
\end{proposition}

\Cref{eq:e,eq:i,eq:u} are restatements of \cite[Proposition 3.1]{a18}. \Cref{eq:v} is a generalization of the (ket) rule given in \cite[Figure 3]{v21} and can be derived from \cref{eq:i} as follows:
\begin{align*}
	\sum_y\ket{\Psi(y)} &= \sum_y\left[\frac{1}{2}\sum_{x,z}(-1)^{z(x + (y + f))}\ket{\Psi(y)}\right] & \text{By \cref{eq:i} where $x, z\notin FV(f)\cup FV(\Psi)$} \\
		&= \sum_x\left[\frac{1}{2}\sum_{y,z}(-1)^{z(y + (x + f))}\ket{\Psi(y)}\right] & \text{Basic arithmetic} \\
		&= \sum_{x}\ket{\Psi(x + f)} & \text{By \cref{eq:i}} \\
		&= \sum_{y}\ket{\Psi(y + f)} & \text{Since $y\notin FV(f)\cup FV(\Psi)$}
\end{align*}
The first equality above uses the instance $\sum_{x,z}(-1)^{z(x + f')}\ket{\Psi(y)} = 2\ket{\Psi(y)}$ of \cref{eq:i}, where $f'=y+f$ and $\ket{\Psi(y)}$ is viewed as a path sum with zero occurrences of the free variable $x$.

\begin{example}
	Consider the dimensionless path sum $\frac{1}{\sqrt{2}}\sum_y i^y$. By \cref{eq:u} it follows that $\frac{1}{\sqrt{2}}\sum_y i^y = \omega.$ Hence, \Cref{eq:u} symbolically encodes the fact that $\omega = \frac{1 + i}{\sqrt{2}}$.
\end{example}

\paragraph{Relationship to post-selected circuits}
As with the ZX-calculus and variants, path sum expressions correspond naturally to circuits with ancillas and postselection. Given a path sum expression of a linear operator $\Psi:\C^{2^m}\rightarrow\C^{2^n}$ of the form of \cref{eq:pathsum1}, a circuit implementing $\Psi$ up to a constant scalar factor can be achieved through postselection as follows. First, prepare $k$ ancillary qubits in the state $\frac{1}{\sqrt{2}^k}\sum_{\vec{y}\in\Z_2^k}\ket{\vec{y}}$ by applying Hadamard gates to the $\ket{0}^{\otimes k}$ state. The symbolic state is then prepared up to some garbage $\ket{g(\vec{x},\vec{y})}$ via the unitary transformation
\[
	\Psi_{PS}:\ket{\vec{x}}\otimes \ket{\vec{y}} 
		\mapsto e^{2\pi iP(\vec{x},\vec{y})}\ket{f(\vec{x},\vec{y})}\otimes \ket{g(\vec{x},\vec{y})}.
\]
Finally, the garbage is discarded by postselecting $\hgate^{\otimes m+k-n}\ket{g(\vec{x},\vec{y})} = \mathcal{N}'\sum_{\vec{z}}(-1)^{\vec{z}\cdot g(\vec{x},\vec{y})}\ket{\vec{z}}$ on $\vec{z}=\vec{0}$. 

Since postselected quantum circuits are believed to be strictly more powerful than non-postselected circuits \cite{a05}, the rest of this paper focuses on the question of synthesizing unitary circuits for path sums representing unitary transformations, up to normalization.

% ------------------------------------------------------------------
\section{Unitarity testing}
\label{sec:unitarity}

We are interested in the synthesis of \emph{unitary} quantum circuits implementing a path sum. As a path sum may represent an arbitrary linear operator, a natural question to ask is whether a given path sum represents a unitary transformation, and hence can be extracted to a unitary circuit. We call this the \emph{unitarity testing} problem for path sums and formulate it as a decision problem below.

\begin{definition}[UNITARY]
	\emph{UNITARY} is the set of path-sums $\ket{\Psi}$ where $\Psi$ is a unitary transformation.
\end{definition}

The unitarity problem is clearly decidable since we can always explicitly compute a matrix representation of $\Psi$ from a path sum $\ket{\Psi}$. However, since the size of the corresponding matrix is exponential in $n$, this solution is not efficient. As we show in this section, one should not hope for an efficient solution in general.

Recall that the complexity class \comp{co-NP} consists of the decision problems whose complement belongs to \comp{NP}, and hence is widely believed to be intractable. A canonical complete problem for \comp{co-NP} is the tautology problem, recognizing the set of propositional formulas over the connectives $\s{\neg, \land,\lor}$ which are satisfied by every variable assignment. We view a propositional formula in $n$ distinct free variables as a function $\Z_2^n \rightarrow \Z_2$, using the standard interpretation $\neg x := 1 + x$, $x\land y := xy$ and $x\lor y := x + y - xy$. The application of $\varphi$ to some $\vec{x}\in\Z_2^n$ is denoted by $\varphi(\vec{x})$.

\begin{definition}[Tautology]
A propositional formula $\varphi$ in $n$ variables is a \emph{tautology} if $\varphi(\vec{x})=1$ for every $\vec{x}\in\Z_2^n$, written $\varphi\equiv 1$.
\end{definition}

\begin{definition}[TAUT]
	\emph{TAUT} is the set of all propositional formulas that are tautologies.
\end{definition}

\begin{theorem}[Karp's 21 NP-complete problems]
	The TAUT problem is \textbf{co-NP}-complete.
\end{theorem}

To reduce TAUT to UNITARY, our goal is to encode a propositional formula $\varphi$ as a path sum representing the linear operator $\Phi\ket{\vec{x}} = \varphi(\vec{x})\ket{\vec{x}}$ which is the identity if $\varphi \equiv 1$, and non-unitary otherwise. To do so we establish an encoding of $\varphi$ as a dimensionless path sum of the form $\varphi(\vec{x}) = \mathcal{N}\sum_{\vec{y}}(-1)^{P(\vec{x}, \vec{y})},$
where $P$ is a multilinear Boolean polynomial.

It can readily be observed that $x = 2^{-1}\sum_{y\in\Z_2}(-1)^{y(1 + x)}$ for any $x\in\Z_2$. This gives an immediate encoding of any propositional formula in a path sum by extending the lifting discussed in \cref{sec:pathsum} to propositional negation and disjunction via the equations $\overline{\neg \varphi} = 1 - \overline{\varphi}$, $\overline{\varphi\lor \psi} = \overline{\varphi} + \overline{\psi} - \overline{\varphi}\cdot\overline{\psi},$
and then setting
\[
	\varphi(\vec{x}) = 2^{-1}\sum_{y}(-1)^{y(1 + \overline{\varphi}(\vec{x}))}.
\]
However, the lifting of a propositional formula $\varphi$ may generally have size exponential in the size of $\varphi$. To obtain a polynomial size encoding we rely on the \emph{Tseytin transformation} \cite{t83}.

Given two propositional formulas $\varphi$ and $\psi$, we write $\varphi \leftrightarrow \psi$ for the logical equality of $\varphi$ and $\psi$, which is satisfied by an assignment $\vec{x}$ if and only if $\varphi(\vec{x})=\psi(\vec{x})$. The Tseytin transformation takes a propositional formula $\varphi$ with $k$ distinct subterms and returns an equisatisfiable conjunction of at most $O(k)$ constant-depth formulas by assigning a fresh propositional variable to the value of each subterm. For instance, given a propositional formula $\varphi = x_1 \land (x_2 \lor \neg x_3)$, the Tseytin transformation of $\varphi$, denoted $\mathcal{T}(\varphi)$, is
\begin{align*}
	\mathcal{T}(\varphi) = z_1 \land (z_1 \leftrightarrow x_1 \land z_2) \land (z_2 \leftrightarrow x_2 \lor z_3) \land (z_3 \leftrightarrow \neg x_3).
\end{align*}
Note that $FV(\varphi) \subseteq FV(\mathcal{T}(\varphi))$ and that the satisfying assignments of $\varphi$ and $\mathcal{T}(\varphi)$ are in a 1-to-1 correspondence and agree on $FV(\varphi)$.

Given a propositional formula $\varphi$, we can encode the Tseytin transformation $\mathcal{T}(\varphi)$ of $\varphi$ in a dimensionless sum over the free variables of $\varphi$ using the following encoding of logical equality
\[
	(\varphi\leftrightarrow\psi)(\vec{x}) = \sum_{y}(-1)^{y\overline{\varphi}(\vec{x}) + y\overline{\psi}(\vec{x})}.
\]
Note that for a clause of the Tseytin transformation $z \leftrightarrow \varphi$ where $\varphi$ has constant depth, $\overline{\varphi}$ has constant size. If we denote the clauses of $\mathcal{T}(\varphi)$ by $z_1\leftrightarrow c_1,\dots, z_k\leftrightarrow c_k$, we may encode $\mathcal{T}(\varphi)$ as a polynomial-size sum by taking the product of each clause and distributing over the summations:
\begin{align*}
	\mathcal{T}(\varphi)(\vec{x},\vec{z}) = z_1\prod_i (z_i\leftrightarrow c_i)(\vec{x})  &= 2^{-1}\sum_{y}(-1)^{y(1 + z_1)}\prod_i 2^{-1}\sum_{y_i}(-1)^{y_i(z_i + \overline{c_i}(\vec{x}))} \\  &= 2^{-(k+1)}\sum_y\sum_{\vec{y}\in\Z_2^k} (-1)^{y(1 + z_1) + \sum_i y_i(z_i + \overline{c_i}(\vec{x}))}.
\end{align*}
Finally, since the satisfying assignments of $\varphi$ and $\mathcal{T}(\varphi)$ are in a 1-to-1 correspondence, we see that
\[
	\varphi(\vec{x}) = \sum_{\vec{z}}\mathcal{T}(\varphi)(\vec{x},\vec{z}) = 2^{-(k+1)}\sum_y\sum_{\vec{y}\in\Z_2^k}\sum_{\vec{z}\in\Z_2^k} (-1)^{y(1 + z_1) + \sum_i y_i(z_i + \overline{c_i}(\vec{x}))}.
\]

\begin{proposition}\label{prop:encoding}
Let $\varphi$ be a propositional formula in $n$ variables and let $\mathcal{T}(\varphi) = z_1\land(\bigwedge_{i=1}^k z_i\leftrightarrow c_i)$. Then for any $\vec{x}\in\Z_2^n$,
\begin{align*}
	\varphi(\vec{x}) &= 2^{-(k+1)}\sum_y\sum_{\vec{y}\in\Z_2^k}\sum_{\vec{z}\in\Z_2^k} (-1)^{y(1 + z_1) + \sum_i y_i(z_i + \overline{c_i}(\vec{x}))}
\end{align*}
where the sum on the right hand size has size polynomial in $k$.
\end{proposition}

\begin{remark}
The encoding of $\varphi$ in \cref{prop:encoding} is interesting because it gives a polynomial-size expression \emph{in the same variables as $\varphi$}. This is in contrast to the propositional Tseytin transformation which gives an encoding over a superset of free variables, and hence only remains equi-satisfiable. In particular, $\mathcal{T}(\cdot)$ does not preserve tautologies, whilst our encoding does when viewed as a $\{0,1\}$-valued function.
\end{remark}

Given the encoding of $\varphi$ above, we can now prove \textbf{co-NP}-hardness of the unitarity testing problem by a reduction from TAUT.

\begin{theorem}
	The unitarity testing problem is \textbf{co-NP}-hard
\end{theorem}

\begin{proof}
	By many-one reduction from TAUT to UNITARY. Given a propositional formula $\varphi$ in $n$ variables, define $\Psi:\Z_2^n\rightarrow \Z_2^n$ to be the linear operator given by $\Psi\ket{\vec{x}} = \varphi(\vec{x})\ket{\vec{x}}.$
	By \cref{prop:encoding}, $\Psi$ admits the following representation as a path sum
	\[
		\Psi\ket{\vec{x}} = 2^{-(k+1)}\sum_y\sum_{\vec{y}\in\Z_2^k}\sum_{\vec{z}\in\Z_2^k} (-1)^{y(1 + z_1) + \sum_i y_i(z_i + \overline{c_i}(\vec{x}))}\ket{\vec{x}}
	\]
	which has size polynomial in the number of subterms of $\varphi$. Hence all that remains is to show that $\Psi$ is unitary 
	if and only if $\varphi$ is a tautology.
	If suffices to observe that, for any $\vec{x}\in\Z_2^n$, $\Psi\ket{\vec{x}} = \varphi(\vec{x})\ket{\vec{x}} = \ket{\vec{x}}$ if $\varphi(\vec{x}) = 1$ and $\vec{x}\in\Z_2^n$, $\Psi\ket{\vec{x}} = \varphi(\vec{x})\ket{\vec{x}} = 0$ otherwise. In particular, if $\varphi(\vec{x}) = 1$ for all $\vec{x}\in\Z_2$, then
	$\Psi= I_n$. Otherwise, there exists $\vec{x}\in\Z_2$
	such that $\Psi\ket{\vec{x}} = 0$ and hence $\Psi$ is non-unitary,
	as required.
\end{proof}

% ------------------------------------------------------------------
\section{Clifford synthesis}
\label{sec:stabilizer}

In this section we look at the problems of synthesis and unitarity testing in the restricted case of Clifford operations. The synthesis of Clifford circuits has applications both to randomized benchmarking, as well as to the design and analysis of error correction circuits. We first review the definition of the Clifford group.

\begin{definition}[Pauli group]
	The $n$-qubit \emph{Pauli group} $\mathcal{P}_n$ is the group of $n$-fold tensor products of Pauli operators $\{I, X, Y, Z\}$.
\end{definition}

\begin{definition}[Clifford group]
	The $n$-qubit \emph{Clifford group} is the group $\mathcal{C}_n = \{U\in U_{2^n} \mid U\mathcal{P}_nU^\dagger = \mathcal{P}_n\}$.
\end{definition}

A well-known consequences of the Gottesman-Knill theorem is the fact that, up to global phases, the Clifford group is generated by $\{\hgate, \pgate, \cxgate\}$. We may use this fact to give a convenient path sum representation of Clifford operations.

\begin{proposition}
	Every Clifford operator $\Psi:\C^{2^n}\rightarrow \C^{2^n}$ can be written as a sum of the form
   	 \begin{equation}\label{eq:stabilizer}
		\Psi\ket{\vec{x}} = \frac{\omega^l}{\sqrt{2^k}}\sum_{\vec{y}\in\Z_2^k}i^{L(\vec{x}, \vec{y})}(-1)^{Q(\vec{x}, \vec{y})}\ket{f(\vec{x}, \vec{y})}
   	 \end{equation}
	where $\omega=e^{2\pi i/8}$, $l\in\Z_8$, $L:\Z_2^{n}\times\Z_2^{k}\rightarrow \Z_4$ is linear, 
	$Q:\Z_2^{n}\times\Z_2^{k}\rightarrow \Z_2$ is pure quadratic, and 
	$f:\Z_2^{n}\times\Z_2^{k}\rightarrow \Z_2$ is affine.
\end{proposition}

\begin{proof}
	As the Clifford group generators $\cxgate, \pgate$, and $\hgate$ can be written in the form of \cref{eq:stabilizer}, it only remains to show that the composition of two such sums can be written in the form of \cref{eq:stabilizer}. It suffices to note that substitution of a variable with an affine Boolean expression does not increase the degree of $Q$ or $f$, while substitution in $L$ produces a quadratic form with degree $2$ terms divisible by $2$.
\end{proof}

We call an expression of the form of \cref{eq:stabilizer} a \emph{Clifford} path sum. In the context of \emph{stabilizer states} --- states $\ket{\psi} = C\ket{\vec{0}}$ for some $C\in\mathcal{C}_n$ --- this representation is well-known by various names, including the \emph{quadratic form expansion} \cite{bh21} and the \emph{affine representation} \cite{dm03,n10}. Through the inclusion of free parameters we can represent stabilizer states, Clifford unitaries, or Clifford circuits with ancillas in this form.

We next define a \emph{normal form} for Clifford path sums which will be useful for circuit synthesis.

\begin{definition}[Normal form]
	A Clifford path sum for $\Psi$ is in \emph{normal form} if, up to a reordering of qubits,
   	 \begin{equation}\label{eq:normal}
		\Psi\ket{\vec{x}} = \frac{\omega^l}{\sqrt{2^k}}\sum_{\vec{y}\in\Z_2^k}i^{L(\vec{x}, \vec{y})}(-1)^{Q(\vec{x}, \vec{y})}\ket{\vec{y}}\otimes\ket{f(\vec{x}, \vec{y})},
   	 \end{equation} where $l\in\Z_8$, 
	$L:\Z_2^{n}\times\Z_2^{k}\rightarrow \Z_4$ is linear, 
	$Q:\Z_2^{n}\times\Z_2^{k}\rightarrow \Z_2$ is pure quadratic, and 
	$f:\Z_2^{n}\times\Z_2^{k}\rightarrow \Z_2$ is affine.
\end{definition}

The normal form above corresponds to re-writing the sum over a minimal set of vectors spanning the affine subspace of $\Z_2^n$ given by $\{f(\vec{x},\vec{y}) \mid \vec{y}\in\Z_2^k\}$. The following proposition states that \cref{eq:e,eq:i,eq:u,eq:v} suffice to re-write a unitary Clifford path sum into normal form.

\begin{proposition}\label{prop:normal}
	Let $\ket{\Psi}$ be a Clifford path sum.
	There exists a re-writing procedure which will
	terminate with $\ket{\Psi}$ in normal form
	if $\Psi$ is unitary and runs in time 
	polynomial in the size of $\ket{\Psi}$.
\end{proposition}
\begin{proof}
For each path variable $y_i$, if there exists $j$ such that $f_j(\vec{x},\vec{y}) = y_i \oplus f'(\vec{x},\vec{y})$, \cref{eq:v} can be applied to substitute $y_i$ with $y_i\oplus f'(\vec{x},\vec{y})$. If no such $j$ exists, either $\Psi$ is unitary and one of \cref{eq:e,eq:i,eq:u} necessarily applies to eliminate $y_i$  \cite{a18}, or no rule applies and $\Psi$ is non-unitary.
\end{proof}

\begin{remark}
	\Cref{prop:normal} also holds for non-square $\Psi:\C^{2^m}\rightarrow \C^{2^n}$ with $m\leq n$ so long as $\Psi$ is an \emph{isometry} --- that is, if $\Psi$ corresponds to a Clifford circuit with some ancillas or fixed inputs.
\end{remark}

\begin{corollary}
	The unitarity testing problem for Clifford path sums is in \textbf{P}.
\end{corollary}

\begin{proof}
	Given a Clifford path sum $\ket{\Psi}$, we can construct a path sum representation of $\Psi^\dagger$ 
	efficiently using $\eta$, $\epsilon$, and negating $L$. Then by \cref{prop:normal} we can normalize the path sum representations 
	of $\Psi\Psi^\dagger$ and $\Psi^\dagger\Psi$ each in polynomial time. If either fails to produce a normal form,
	then one of $\Psi\Psi^\dagger$ or $\Psi^\dagger\Psi$ is non-unitary and hence $\Psi$ is non-unitary.
	If both are reduced to normal form, it suffices to observe that we can check whether
	a Clifford path sum in normal form represents the identity 
	transformation in polynomial time.
\end{proof}

It is now straightforward to compute a circuit implementing a (unitary) Clifford path sum from its normalized form. If we decompose $f$, $L$, and $Q$ as $f(\vec{x},\vec{y}) = f_x(\vec{x}) + f_y(\vec{y}) + \vec{b}$, $L(\vec{x}, \vec{y}) = L_x(\vec{x}) + L_y(\vec{y})$, and $Q(\vec{x}, \vec{y}) = Q_x(\vec{x}) + Q_y(\vec{y}) + \sum_{i=1}^ky_iR_i(\vec{x})$ then the normal form can be written as the following sequence of transformations:
\begin{align*}
	\ket{\vec{x}} &\mapsto \omega^li^{L_x(\vec{x})}(-1)^{Q_x(\vec{x})}\ket{\vec{x}} \\ 
	\ket{\vec{x}} &\mapsto \ket{R(\vec{x})} \ket{f_x(\vec{x})} \\
	\ket{R(\vec{x})}\ket{f_x(\vec{x})} &\mapsto \frac{1}{\sqrt{2^k}}\sum_{\vec{y}\in\Z_2^k} (-1)^{\sum_i y_iR_i(\vec{x})}\ket{\vec{y}}\ket{f_x(\vec{x})} \\
	\ket{\vec{y}} \ket{f_x(\vec{x})} &\mapsto \ket{\vec{y}} \ket{f_x(\vec{x}) + f_y(\vec{y}) + \vec{b}} \\
	\ket{\vec{y}}\ket{f(\vec{x},\vec{y})} &\mapsto i^{L_y(\vec{y})}(-1)^{Q_y(\vec{y})}\ket{\vec{y}}\ket{f(\vec{x},\vec{y})}
\end{align*}
This gives a circuit of the form $U(\hgate^{\otimes k}\otimes \igate_{n-k})V$ where $U$ and $V$ are generalized permutations contained in the Clifford group. Moreover, $\ket{\vec{x}} \mapsto \ket{R(\vec{x})} \ket{f_x(\vec{x})}$ is the only operator which may be non-unitary, and in particular is unitary if and only if $\Psi$ is. Note that the unitarity of $\Psi$ hence forces $\{R_i\}$ to be linearly independent and for $R_i$ to be non-zero. This is summarized in \cref{alg:clifford}.

\begin{algorithm}[h]
\caption{Clifford synthesis algorithm}
\label{alg:clifford}
\begin{enumerate}
	\itemsep0em
	\item Normalize $\ket{\Psi}$ in the form $\frac{\omega^l}{\sqrt{2^k}}\sum_{\vec{y}\in\Z_2^k}
		i^{L(\vec{x}, \vec{y})}(-1)^{Q(\vec{x}, \vec{y})}\ket{\vec{y}}\otimes\ket{f(\vec{x}, \vec{y})}$ up to qubit reordering.
	\item Decompose $f$, $L$, and $Q$ as $f(\vec{x},\vec{y}) = f_x(\vec{x}) + f_y(\vec{y}) + \vec{b}$, $L(\vec{x}, \vec{y}) = L_x(\vec{x}) + L_y(\vec{y})$ and $Q(\vec{x}, \vec{y}) = Q_x(\vec{x}) + Q_y(\vec{y}) + \sum_{i}y_iR_i(\vec{x})$ where each $R_i$ is linear.
	\item Synthesize circuits for the following linear transformations:
	\begin{itemize}
		\item $D\ket{\vec{x}} = i^{L_x(\vec{x})}(-1)^{Q_x(\vec{x})}\ket{\vec{x}}$
		\item $U \ket{\vec{x}} = \ket{R(\vec{x})}\ket{f_x(\vec{x})}$
		\item $V \ket{\vec{y}}\ket{f_x(\vec{x})} = \ket{\vec{y}}\ket{f_x(\vec{x}) + f_y(\vec{y}) + \vec{b}}$
		\item $P\ket{\vec{y}}\ket{f(\vec{x},\vec{y})} = i^{L_y(\vec{y})}(-1)^{Q_y(\vec{y})}\ket{\vec{y}}\ket{f(\vec{x},\vec{y})}$
	\end{itemize}
	\item Return $\omega^l PV(\hgate^{\otimes k}\otimes \igate_{n-k})UD$ with qubits appropriately reordered.
\end{enumerate}
\end{algorithm}

\begin{theorem}\label{thm:clifford}
	Let $\Psi:\C^{2^n}\rightarrow \C^{2^n}$ be expressed as a Clifford path sum. 
	If $\Psi$ is unitary, then \cref{alg:clifford} produces a circuit over $\{\omega, \cxgate, \xgate, \czgate, \pgate, \hgate\}$ implementing $\Psi$
	in time polynomial in the size of the expression. Moreover, this circuit can be written up to global phase as an $8$-stage circuit of the form
	\[
		\pgate\cdot\czgate\cdot\cxgate\cdot\hgate
		\cdot\cxgate\cdot\xgate\cdot\czgate\cdot\pgate
	\]
\end{theorem}

\begin{proof}
That $\Psi = \omega^l PV(\hgate^{\otimes k}\otimes \igate_{n-k})UD$ follows by an easy calculation.

By \cref{prop:normal}, $\Psi$ can be written up to a permutation of qubits in normal form in polynomial time.
Since $L_x$ and $L_y$ are linear, and $Q_x$ and $Q_y$ are pure quadratic, $D$ and $P$ can each be synthesized
using a single stage each of $\pgate$ and $\czgate$ gates --- one $\pgate^m$ gate for each non-zero term of $L_{\{x,y\}}$
and one $\czgate$ gate for each non-zero term of $Q_{\{x,y\}}$. 
Likewise, since $f_y(\vec{y})$ is linear, $V$ can be synthesized in time polynomial in $n$ using a single stage each 
of $\cxgate$ and $\xgate$ gates --- one gate for each non-zero entry of $f_y(\vec{y})$ and $\vec{b}$. 
Finally, $U\ket{\vec{x}} = \ket{R(\vec{x})}\otimes \ket{f_x(\vec{x})}$ can be synthesized over $\{\cxgate\}$ in
polynomial time using Gaussian elimination if and only if $U$ is invertible. 
Moreover, since 
\[
	U=\omega^{-l}(\hgate^{\otimes k}\otimes  \igate_{n-k})V^\dagger P^\dagger\Psi D^\dagger,
\] 
it follows that $U$ is invertible if and only if $\Psi$ is unitary.
\end{proof}

We give a diagrammatic presentation of \cref{thm:clifford} showing the circuit schematically below.

{\small
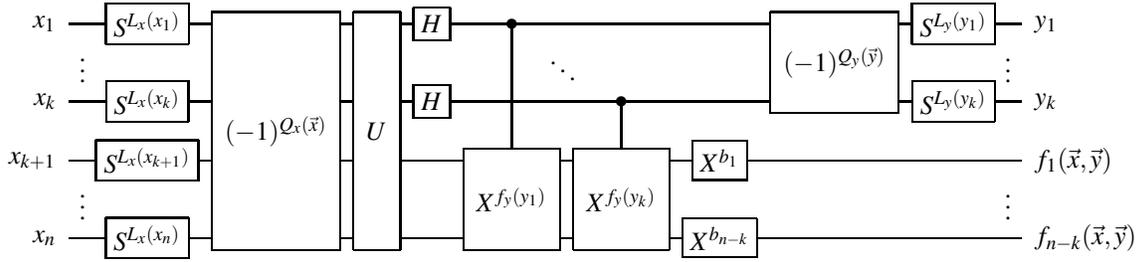
\[\hspace*{-3em}
	\Qcircuit @C=.5em @R=.7em {
		\lstick{x_1} & \qw & \gate{S^{L_{x}(x_1)}} & \multigate{5}{(-1)^{Q_x(\vec{x})}} & \multigate{5}{U} & \gate{H}& \ctrl{3} & \qw & \qw & \multigate{2}{(-1)^{Q_y(\vec{y})}} & \gate{S^{L_{y}(y_1)}} & \qw & \rstick{y_1} \qw \\
		& \vdots & & & & & \rstick{\;\;\;\;\ddots} & & & & & \vdots \\
		\lstick{x_k} & \qw & \gate{S^{L_{x}(x_k)}} & \ghost{(-1)^{Q_x(\vec{x})}} & \ghost{U} & \gate{H} & \qw & \ctrl{1} & \qw & \ghost{(-1)^{Q_y(\vec{y})}} & \gate{S^{L_{y}(y_k)}} & \qw & \rstick{y_k} \qw \\
		\lstick{x_{k+1}} & \qw & \gate{S^{L_{x}(x_{k+1})}} & \ghost{(-1)^{Q_x(\vec{x})}} & \ghost{U} & \qw & \multigate{2}{X^{f_y(y_1)}} & \multigate{2}{X^{f_y(y_k)}} & \gate{X^{b_1}} & \qw & \qw & \qw & \rstick{f_1(\vec{x},\vec{y})} \qw \\
		& \vdots & & & & & & & & & & \vdots \\
		\lstick{x_n} & \qw & \gate{S^{L_{x}(x_n)}} & \ghost{(-1)^{Q_x(\vec{x})}} & \ghost{U} & \qw & \ghost{X^{f_y(y_1)}} & \ghost{X^{f_y(y_k)}} & \gate{X^{b_{n-k}}} & \qw & \qw & \qw & \rstick{f_{n-k}(\vec{x},\vec{y})} \qw
	}
\]
}

\paragraph{Discussion}
In \cite{mr18} a $7$ stage decomposition of the Clifford group of the form 
$\pgate\cdot\czgate\cdot\textsc{c}\cdot\hgate\cdot\textsc{c}\cdot\czgate\cdot\pgate$
was given, where $\textsc{c}$ is circuit implementing an affine permutation. 
As affine permutations require both $\cxgate$ and $\xgate$ gates to implement without 
ancillas --- and moreover $\xgate$ can not be written in the form 
$\pgate\cdot\czgate\cdot\cxgate\cdot\hgate\cdot\cxgate\cdot\czgate\cdot\pgate$ --- our 
projective decomposition reduces the equivalent
$9$-stage projective decomposition of \cite{mr18} to $8$ stages.

It can also be observed that with a minor modification, \cref{alg:clifford} suffices to synthesize Clifford circuits with ancillas, including circuits for preparing stabilizer states. In particular, if $\Psi:\C^{2^m}\rightarrow \C^{2^n}$ is an isometric Clifford path sum with $m \leq n$, the only modification needed is in the synthesis of $U\ket{\vec{x}}= \ket{R(\vec{x})}\otimes \ket{f_x(\vec{x})}$. If indeed $m\leq n$, then a Clifford circuit with ancillas exists and can be synthesized if and only if $\{R_i\}\cup\{(f_x)_i\}$ contains $m$ linearly independent (row) vectors. This produces a $5$-stage circuit of the form $\hgate\cdot\cxgate\cdot \xgate\cdot\czgate\cdot\pgate$ for the preparation of an arbitrary stabilizer state up to global phase. This stabilizer state decomposition was previously given in \cite{n10}.

\begin{corollary}
	A Clifford normal form
	$
		\ket{\vec{x}} \mapsto \frac{\omega^l}{\sqrt{2^k}}\sum\nolimits_{\vec{y}\in\Z_2^k}i^{L(\vec{x}, \vec{y})}(-1)^{Q(\vec{x}, \vec{y})}\ket{\vec{y}}\otimes\ket{f(\vec{x}, \vec{y})}
	$
	 from $m$ to $n\geq m$ qubits
	can be implemented with Clifford gates and ancillas initialized in the $\ket{0}$ state \emph{if and only if}
	\[
		\text{rank}(\{R_i\} \cup \{(f_x)_i\}) = m.
	\]
	The circuit is shown schematically below:

{\small
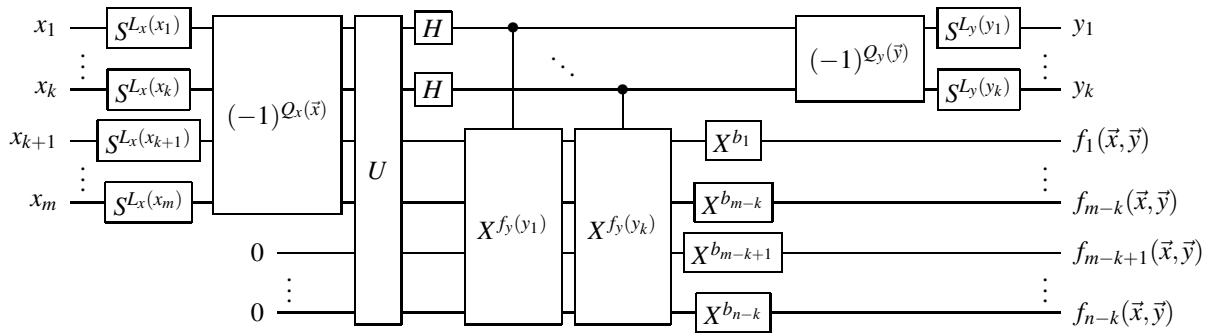
\[\hspace*{-3em}
	\Qcircuit @C=.5em @R=.4em {
		\lstick{x_1} & \qw & \gate{S^{L_{x}(x_1)}} & \multigate{5}{(-1)^{Q_x(\vec{x})}} & \multigate{8}{U} & \gate{H}& \ctrl{3} & \qw & \qw & \multigate{2}{(-1)^{Q_y(\vec{y})}} & \gate{S^{L_{y}(y_1)}} & \qw & \qw & \rstick{y_1} \qw \\
		& \vdots & & & & & \rstick{\;\;\;\;\ddots} & & & & & & \vdots \\
		\lstick{x_k} & \qw & \gate{S^{L_{x}(x_k)}} & \ghost{(-1)^{Q_x(\vec{x})}} & \ghost{U} & \gate{H} & \qw & \ctrl{1} & \qw & \ghost{(-1)^{Q_y(\vec{y})}} & \gate{S^{L_{y}(y_k)}} & \qw & \qw & \rstick{y_k} \qw \\
		\lstick{x_{k+1}} & \qw & \gate{S^{L_{x}(x_{k+1})}} & \ghost{(-1)^{Q_x(\vec{x})}} & \ghost{U} & \qw & \multigate{5}{X^{f_y(y_1)}} & \multigate{5}{X^{f_y(y_k)}} & \gate{X^{b_1}} & \qw & \qw & \qw & \qw & \rstick{f_1(\vec{x},\vec{y})} \qw \\
		& \vdots & & & & & & & & & & & \vdots \\
		\lstick{x_m} & \qw & \gate{S^{L_{x}(x_m)}} & \ghost{(-1)^{Q_x(\vec{x})}} & \ghost{U} & \qw & \ghost{X^{f_y(y_1)}} & \ghost{X^{f_y(y_k)}} & \gate{X^{b_{m-k}}} & \qw & \qw & \qw & \qw & \rstick{f_{m-k}(\vec{x},\vec{y})} \qw \\
		& & & \lstick{0} & \ghost{U} & \qw & \ghost{X^{f_y(y_1)}} & \ghost{X^{f_y(y_k)}} & \gate{X^{b_{m-k+1}}} & \qw & \qw & \qw & \qw & \rstick{f_{m-k+1}(\vec{x},\vec{y})} \qw \\
		& & & \hspace{1em}\vdots & & & & & & & & & \vdots \\
		& & & \lstick{0} & \ghost{U} & \qw & \ghost{X^{f_y(y_1)}} & \ghost{X^{f_y(y_k)}} & \gate{X^{b_{n-k}}} & \qw & \qw & \qw & \qw & \rstick{f_{n-k}(\vec{x},\vec{y})} \qw
	}
\]
}
\end{corollary}

% ------------------------------------------------------------------
\section{General synthesis}
\label{sec:general}

We now consider the more challenging problem of synthesizing a unitary circuit from an arbitrary path sum. Our method attempts to iteratively reduce the number of summed variables in a path sum by alternately applying generalized permutations and Hadamard gates to the symbolic state. Recall that a \emph{(unitary) generalized permutation} is a permutation matrix whose nonzero entries are elements of $\mathbb{T} = \{z\in \C \mid |z| = 1\}$. The generalized permutations are generated by the gates
\[
\Lambda_k(X):\ket{\vec{x}}\ket{y} \mapsto \ket{\vec{x}}\ket{y \oplus \prod_i x_i}
\quad \mbox{ and } \quad
\Lambda_k(R_Z(\theta)):\ket{\vec{x}}\mapsto e^{2\pi i \theta\prod_i x_i}\ket{\vec{x}}.
\]
Together with the Hadamard gate this forms an exactly universal set as it includes every single-qubit unitary along with the $\cxgate$ gate.

The following fact forms the basis of our synthesis algorithm. It gives a condition on a path sum which allows a summed variable to be eliminated by multiplication with a Hadamard gate.
\begin{proposition}\label{prop:hadamard}
	Let $\Psi:\C^{2^m}\rightarrow \C^{2^n}$ be a linear operator where
	\[
		\hspace*{6em}\Psi\ket{\vec{x}} = \mathcal{N}\sum_{z\in\Z_2}\sum_{\vec{y}\in\Z_2^{k}}(-1)^{zQ(\vec{x}, \vec{y})}e^{2\pi i P(\vec{x}, \vec{y})}\ket{z}\otimes \ket{f(\vec{x}, \vec{y})}.
	\]
	Then $(\hgate\otimes \igate_{n-1})\Psi\ket{\vec{x}} = \sqrt{2}\mathcal{N}\sum_{\vec{y}\in\Z_2^{k}}e^{2\pi i P(\vec{x}, \vec{y})}\ket{Q(\vec{x}, \vec{y})}\otimes \ket{f(\vec{x}, \vec{y})}.$

\end{proposition}
\begin{proof}
	By \cref{eq:i}, since $(\hgate \otimes \igate_{n-1})\Psi\ket{\vec{x}} = 
		\frac{1}{\sqrt{2}}\mathcal{N}\sum_{z}\sum_{\vec{y}}
		(-1)^{zQ(\vec{x}, \vec{y}) + zz'}e^{2\pi i P(\vec{x}, \vec{y})}\ket{z'}\otimes \ket{f(\vec{x}, \vec{y})}$
\end{proof}

Note that \cref{prop:hadamard} is essentially an inversion of the $\hgate$ gate, $\hgate^\dagger:\sum_z(-1)^{xz}\ket{z} \mapsto \ket{x}$.
We say that a variable $z$ is \emph{reducible} in the path sum $\ket{\Psi}$ if $\ket{\Psi}$ is in the form of \cref{prop:hadamard}.

\begin{definition}[Reducible]
	A variable $z$ is \emph{reducible} in an expression of $\Psi:\C^{2^m}\rightarrow \C^{2^n}$ if, up to qubit reordering, it has the form
	\[
		\hspace*{6em}\Psi\ket{\vec{x}} = \mathcal{N}\sum_{z\in\Z_2}\sum_{\vec{y}\in\Z_2^{k}}(-1)^{zQ(\vec{x}, \vec{y})}e^{2\pi i P(\vec{x}, \vec{y})}\ket{z}\otimes \ket{f(\vec{x}, \vec{y})}.
	\]
\end{definition}

At a high level, our algorithm proceeds by attempting to synthesize a generalized permutation which will leave some path variable reducible. If the process terminates with remaining summed variables, or an unsynthesizeable ground term $e^{2\pi i P(\vec{x})}\ket{f(\vec{x})}$, the algorithm fails to produce a circuit. \Cref{alg:unitary} gives the high-level algorithm in pseudo-code.

\begin{algorithm}[h]
\caption{General path sum synthesis algorithm}
\label{alg:unitary}
\begin{enumerate}
	\itemsep0em
	\item Set $C$ to the empty circuit and normalize $\ket{\Psi}$ using \cref{eq:e,eq:i,eq:u}
	\item For each remaining path variable $y$ in $\ket{\Psi}$
	\begin{enumerate}
		\item If there exists a generalized permutation $U$ such that $y$ is reducible in $U^\dagger\ket{\Psi}$,
		\begin{enumerate}
			\item $\ket{\Psi} \gets (\hgate \otimes \igate_{n-1})U^\dagger\ket{\Psi}$
			\item Append $U(\hgate \otimes \igate_{n-1})$ to $C$
			\item Go to step $1$
		\end{enumerate}
	\end{enumerate}
	\item If path variables remain or $\Psi$ is non-unitary, fail. Otherwise, append $\Psi^\dagger$ to $C$ and return $C$.
\end{enumerate}
\end{algorithm}

Finding such a generalized permutation is highly non-trivial. Our method applies a series of symbolic simplifications, corresponding to $\Lambda_k(X)$ and $\Lambda_k(R_Z(\theta))$ gates, to the term $e^{2\pi i P(\vec{x}, \vec{y})}\ket{f(\vec{x}, \vec{y})}$. If these simplifications fail to leave any variable reducible, we fall back to an exponential-time procedure aimed at computing a substitution of the form in \cref{eq:v} which will make some variable reducible. These heuristics are described in \cref{app:genperm}.

% ------------------------------------------------------------------
\section{Experiments}
\label{sec:applications}

To test the performance and utility of our synthesis methods, we implemented \cref{alg:clifford,alg:unitary} in the \textsc{feynman}\footnote{Available at \href{https://github.com/meamy/feynman}{https://github.com/meamy/feynman}.} software package. In this section, we briefly detail our investigations into applications to the optimization and decompilation of circuits, and to specification-based synthesis. All Clifford circuits synthesized have been checked for correctness using the method of \cite{a18}. For \cref{alg:unitary}, as the method of \cite{a18} often fails to verify circuits extracted using \cref{eq:v}, we instead validated correctness of our synthesis procedure by verifying the individual synthesis steps each on 1000 unitary path sums extracted from randomly generated Clifford+$T$ circuits. \Cref{tab:results} gives some statistics from experiments re-synthesizing random circuits. Random circuits were generated by selecting a given number of gates on a given number of qubits, taken from the $\{\cxgate, \hgate, \pgate\}$ and $\{\cxgate, \hgate, \tgate\}$ gate sets for Clifford and Clifford+$T$, respectively.

\renewcommand{\arraystretch}{0.5}
\begin{table}[h]
\small
\centering
\begin{tabular}{p{10em}rrrrrr} \toprule
	& $n$ & \# gates & \# circuits & avg. time (s) & avg. change (+/-) & success \\ \midrule
Clifford
 & 20 & 500 & 1000 & 0.137 & {\color{red} +19.2\%} & \multicolumn{1}{c}{--} \\
 & 20 & 1000 & 1000 & 0.481 & {\color{green} -12.9\%} & \multicolumn{1}{c}{--} \\ \cmidrule{2-7}
 & 50 & 500 & 1000 & 0.264 & {\color{red} +90.7\%} & \multicolumn{1}{c}{--} \\
 & 50 & 1000 & 1000 & 1.518 & {\color{red} +129.1\%} & \multicolumn{1}{c}{--} \\ \midrule
Clifford+$T$
 & 20 & 100 & 1000 & 0.010 & {\color{red} +48.9\%} & 99.9\%  \\
 & 20 & 200 & 1000 & 0.045 & {\color{red} +93.7\%} & 94.9\% \\ 
 & 20 & 300 & 1000 & 0.097 & {\color{red} +115.9\%} & 74.7\% \\ \cmidrule{2-7}
 & 50 & 100 & 1000 & 0.016 & {\color{red} +33.5\%} & 100.0\%  \\
 & 50 & 200 & 1000 & 0.044& {\color{red} +49.0\%} & 100.0\%  \\
 & 50 & 300 & 1000 & 0.104& {\color{red} +79.4\%} & 99.6\%  \\ \bottomrule
\end{tabular}
\caption{Re-synthesis results for randomly generated circuits on $n$ qubits. Avg. change gives the average percent increase (+) or decrease (-) in the re-synthesized gate count compared to the original circuit. Success gives the percentage of circuits successfully re-synthesized.}
\label{tab:results}
\end{table}
\renewcommand{\arraystretch}{1}

\paragraph{Circuit optimization}

One of the key factors in phase folding optimizations \cite{amm14,nrscm17} is the placement of Hadamard gates. It was shown in \cite{am19} that the $T$-count in a Clifford+$T$ circuit can be upper bounded by $O(hn^2)$, where $h$ is the number of Hadamard layers in the circuit. As our Clifford synthesis algorithm produces circuits with just a single layer of Hadamard gates, it is natural to ask whether we can optimize $T$-count by reducing the number of Hadamard layers in Clifford+$T$ circuits.

We implemented a Clifford sub-circuit normalization method (the \texttt{-clifford} pass in \textsc{feynopt}) using \cref{alg:clifford} to re-synthesize simple greedily chosen Clifford sub-circuits. We tested the effect on $T$-count optimization by applying Clifford normalization together with phase folding and compared it against \cite{kw19} on the benchmark set of \cite{amm14}. In all but 4 benchmarks, the same $T$-count was achieved by normalizing greedily chosen Clifford sub-circuits and applying phase polynomial optimizations. In two of those cases, \texttt{qcla-com}$_{7}$ and \texttt{csla-mux}$_{3}$, our method produced lower $T$-count circuits --- $94$ (down from $95$) and $60$ (down from $62$), respectively. For the other two cases, \texttt{ham15-med} and \texttt{adder}$_8$, our method produced worse results --- $230$ (up from $212$) and $215$ (up from $173$), respectively.

More broadly, we might expect to be able to optimize a circuit by resynthesizing its simplified path sum using the general synthesis algorithm \cref{alg:unitary}. This is often effective when the path sum is simple, as in the resynthesized circuits corresponding to sub-circuits of the \texttt{adder}$_8$ benchmark below, but as the path sum becomes increasingly complex extraction typically performs worse than human designs. We leave it as an avenue for future work to make symbolic synthesis practical for circuit optimization, and in particular to develop effective peephole optimization procedures.

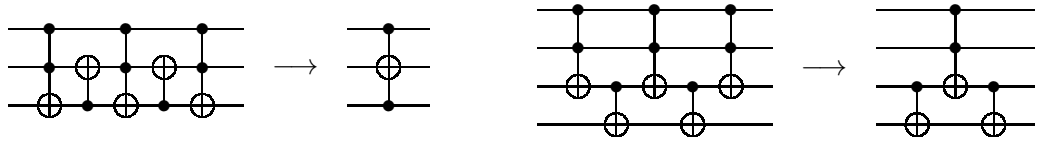
\begin{minipage}{0.43\textwidth}
\[
	\Qcircuit @C=.5em @R=.5em @!R {
		& \qw & \ctrl{2} & \qw & \ctrl{2} & \qw & \ctrl{2} & \qw & \qw  \\
		& \qw & \ctrl{1} & \targ & \ctrl{1} & \targ & \ctrl{1} & \qw & \qw  \\
		& \qw & \targ & \ctrl{-1}& \targ & \ctrl{-1} & \targ & \qw & \qw
	}
	\raisebox{-1.3em}{$\quad\longrightarrow\quad$}
	\Qcircuit @C=.5em @R=.5em @!R {
		& \qw & \ctrl{1} & \qw & \qw  \\
		& \qw & \targ & \qw & \qw  \\
		& \qw & \ctrl{-1} & \qw & \qw
	}
\]
\end{minipage}
\begin{minipage}{0.5\textwidth}
\[
	\Qcircuit @C=.5em @R=.5em @!R {
		& \qw & \ctrl{2} & \qw & \ctrl{2} & \qw & \ctrl{2} & \qw & \qw  \\
		& \qw & \ctrl{1} & \qw& \ctrl{1} & \qw & \ctrl{1} & \qw & \qw  \\
		& \qw & \targ & \ctrl{1}& \targ & \ctrl{1} & \targ & \qw & \qw \\
		& \qw & \qw & \targ & \qw & \targ & \qw & \qw & \qw
	}
	\raisebox{-2em}{$\quad\longrightarrow\quad$}
	\Qcircuit @C=.5em @R=.5em @!R {
		& \qw & \qw & \ctrl{2} & \qw & \qw & \qw  \\
		& \qw & \qw& \ctrl{1} & \qw & \qw & \qw  \\
		& \qw & \ctrl{1}& \targ & \ctrl{1} & \qw & \qw \\
		& \qw & \targ & \qw & \targ & \qw & \qw
	}
\]
\end{minipage}

\paragraph{Decompilation}

An interesting application of our symbolic synthesis algorithm is to the \emph{decompilation} of quantum circuits. Classically, decompilation is the process of translating a program in a low-level language to equivalent high-level source code, typically used for reverse engineering or recompilation. As the gate set targeted by \cref{alg:unitary} is quite high-level, in many cases it can be used to effectively decompile lower-level circuits. This decompilation can potentially help developers to examine the high-level structure of a low-level circuit, and also allow optimizations targeting higher level gate sets to be performed on circuits written over low-level gate sets, such as Clifford+$T$. Below we give some examples of standard circuits from the literature decompiled using \cref{alg:unitary}. The decompiler can be accessed with the \texttt{-decompile} option in \textsc{feynopt}.

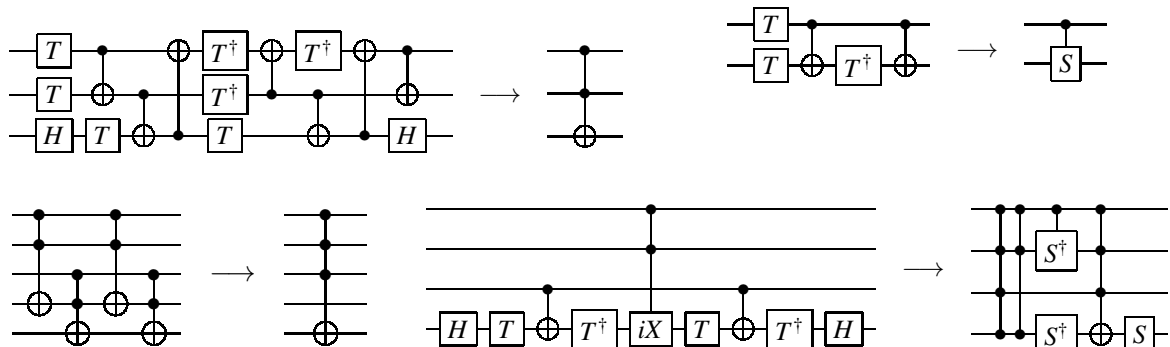
\begin{minipage}[t]{0.5\textwidth}
\small
\[
	\Qcircuit @C=.5em @R=.2em {
		& \qw & \gate{T} & \ctrl{1} & \qw & \targ & \gate{T^\dagger} & \targ & \gate{T^\dagger} & \targ & \ctrl{1} & \qw & \qw  \\
		& \qw & \gate{T} & \targ & \ctrl{1} & \qw & \gate{T^\dagger} & \ctrl{-1} & \ctrl{1} & \qw & \targ & \qw & \qw  \\
		& \qw & \gate{H} & \gate{T} & \targ & \ctrl{-2} & \gate{T} & \qw & \targ & \ctrl{-2} & \gate{H} & \qw & \qw
	}
	\raisebox{-1.7em}{$\quad\longrightarrow\quad$}
	\Qcircuit @C=.5em @R=.8em @!R {
		& \qw & \ctrl{1} & \qw & \qw  \\
		& \qw & \ctrl{1} & \qw & \qw  \\
		& \qw & \targ & \qw & \qw
	}
\]
\end{minipage}
\begin{minipage}[t]{0.5\textwidth}
\small
\[
	\Qcircuit @C=.5em @R=.2em {
		& \qw & \gate{T} & \ctrl{1} & \qw & \ctrl{1} & \qw   \\
		& \qw & \gate{T} & \targ & \gate{T^\dagger} & \targ & \qw 
	}
	\raisebox{-1em}{$\quad\longrightarrow\quad$}
	\Qcircuit @C=.5em @R=.2em @!R {
		& \qw & \ctrl{1} & \qw & \qw  \\
		& \qw & \gate{S} & \qw & \qw
	}
\]
\end{minipage}

\begin{minipage}[t]{0.3\textwidth}
\[
	\Qcircuit @C=.5em @R=.2em @!R {
		& \ctrl{3} & \qw & \ctrl{3} & \qw & \qw \\
		& \ctrl{2} & \qw & \ctrl{2} & \qw & \qw \\
		& \qw & \ctrl{2} & \qw & \ctrl{2} & \qw \\
		& \targ & \ctrl{1} & \targ & \ctrl{1} & \qw \\
		& \qw & \targ & \qw & \targ & \qw
	}
	\raisebox{-2em}{$\quad\longrightarrow\quad$}
	\Qcircuit @C=.5em @R=.2em @!R {
		& \qw & \ctrl{4} & \qw & \qw  \\
		& \qw & \ctrl{3} & \qw & \qw  \\
		& \qw & \ctrl{2} & \qw & \qw  \\
		& \qw & \qw & \qw & \qw  \\
		& \qw & \targ & \qw & \qw
	}
\]
\end{minipage}
\begin{minipage}[t]{0.7\textwidth}
\small
\[
	\Qcircuit @C=.5em @R=.05em @!R {
		& \qw & \qw & \qw & \qw & \ctrl{3} & \qw & \qw & \qw & \qw & \qw \\
		& \qw & \qw & \qw & \qw & \ctrl{2} & \qw & \qw & \qw & \qw & \qw \\
		& \qw & \qw & \ctrl{1} & \qw & \qw & \qw & \ctrl{1} & \qw & \qw & \qw \\
		& \gate{H} & \gate{T} & \targ & \gate{T^\dagger} & \gate{iX} & \gate{T} & \targ & \gate{T^\dagger} & \gate{H} & \qw
	}
	\raisebox{-2.3em}{$\quad\longrightarrow\quad$}
	\Qcircuit @C=.5em @R=.1em @!R {
		& \qw & \ctrl{3} & \ctrl{3} & \ctrl{1} & \ctrl{3} & \qw & \qw  \\
		& \qw & \ctrl{2} & \ctrl{2} & \gate{S^\dagger} & \ctrl{2} & \qw & \qw  \\
		& \qw & \ctrl{1} & \qw & \qw & \ctrl{1} & \qw & \qw  \\
		& \qw & \ctrl{0} & \ctrl{0} & \gate{S^\dagger} & \targ & \gate{S} & \qw
	}
\]
\end{minipage}
The bottom right circuit above is a relative phase Toffoli gate implementation taken from \cite{ar21}. The utility of decompilation is apparant here, as both the fact that it implements a Toffoli up to phase and the exact form of the relative phase can be readily observed from the decompiled circuit.

\paragraph{Specification-based synthesis}

In \cite{a18} it was noted that path sums offer a convenient form of logical specification for many quantum computations, being very close to the ``textbook'' specification. \Cref{alg:unitary} gives a method of synthesizing a circuit directly from such a specification. Such specifications include not only classical reversible functions such as $\ket{x_1x_2x_3} \mapsto \ket{x_1x_2(x_3\oplus x_1x_2)}$, which can be synthesized by existing reversible circuit synthesis methods, but also classical functions ``in the phase,'' up to relative phases, or inside superpositions. We illustrate this by using \cref{alg:unitary} to synthesize the quantum Fourier transform.

Recall that the $n$-qubit quantum fourier transform can be expressed as $QFT_n\ket{\vec{x}} = \frac{1}{\sqrt{2^n}}\sum_{\vec{y}\in\Z_2^n}\omega_{2^n}^{\vec{x}\vec{y}}\ket{\vec{y}}
$
where $\vec{x}\vec{y}$ is the integer product of $\vec{x}$ and $\vec{y}$. 
In the $3$ qubit case, expanding the integer multiplication to a multilinear polynomial we have
\[
	QFT_3\ket{x_1x_2x_3} = \frac{1}{\sqrt{2^3}}\sum_{y_1,y_2,y_3}\omega^{x_3y_3}i^{x_3y_2 + x_2y_3}(-1)^{x_3y_1 + x_2y_2 + x_1y_3}\ket{y_1y_2y_3}
\]
where $y_1$ is reducible, and in particular
\[
	(\hgate\otimes \igate_2)QFT_3\ket{x_1x_2x_3} = \frac{1}{\sqrt{2^2}}\sum_{y_2,y_3}\omega^{x_3y_3}i^{x_3y_2 + x_2y_3}(-1)^{x_2y_2 + x_1y_3}\ket{x_3y_2y_3}.
\]
While neither of $y_2$ or $y_3$ are reducible above, the $\omega^{x_3y_3}$ and $i^{x_3y_2}$ terms can be eliminated by applying controlled-$\tgate^\dagger$ and -$\pgate^\dagger$ gates, respectively, leaving $y_2$ reducible:
\[
	(\Lambda(\pgate^\dagger)\otimes \igate)(\swapgate\otimes \igate)(\igate\otimes \Lambda(\tgate^\dagger))(\swapgate\otimes \igate)(\hgate\otimes \igate_2)QFT_3\ket{x_1x_2x_3} = \frac{1}{\sqrt{2^2}}\sum_{y_2,y_3}i^{x_2y_3}(-1)^{x_2y_2 + x_1y_3}\ket{x_3y_2y_3}.
\]
After eliminating $y_2$, the process repeats for $y_1$, leaving a final permutation to be synthesized.

A $5$ qubit $QFT$ circuit synthesized with our implementation is shown verbatim below, where $R_k:=R_Z(1/2^k)$. We were able to synthesize instances on up to 50 qubits in just seconds on a desktop computer.

{\small
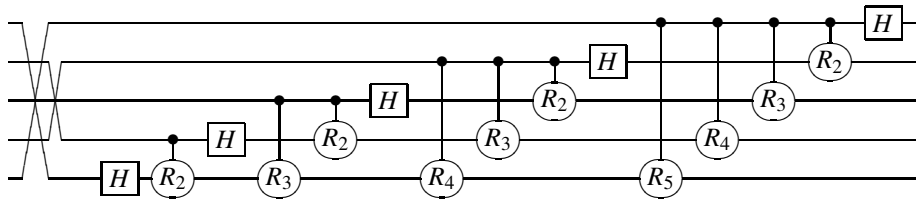
\[
	\Qcircuit @C=.5em @R=.05em @!R {
		& \qw & & \link{4}{-2} & \qw & \qw & \qw & \qw & \qw & \qw & \qw & \qw & \qw & \qw & \qw & \qw & \qw & \ctrl{4} & \ctrl{3} & \ctrl{2} & \ctrl{1} & \gate{H} & \qw \\
		& \qw & \qw & \qw & \link{2}{-1} & \qw & \qw & \qw & \qw & \qw & \qw & \qw & \qw & \ctrl{3} & \ctrl{2} & \ctrl{1} & \gate{H} & \qw & \qw & \qw & \measure{R_2} & \qw & \qw \\
		& \qw & \qw & \qw & \qw & \qw & \qw & \qw & \qw & \qw & \ctrl{2} & \ctrl{1} & \gate{H} & \qw & \qw & \measure{R_2} & \qw & \qw & \qw & \measure{R_3} & \qw & \qw & \qw \\
		& \qw & \qw & \qw & \link{-2}{-1} & \qw & \qw & \qw & \ctrl{1} & \gate{H} & \qw & \measure{R_2} & \qw & \qw & \measure{R_3} & \qw & \qw & \qw & \measure{R_4} & \qw & \qw & \qw & \qw \\
		& \qw & & \link{-4}{-2} & \qw & \qw & \qw & \gate{H} & \measure{R_2} & \qw & \measure{R_3} & \qw & \qw & \measure{R_4} & \qw & \qw & \qw & \measure{R_5} & \qw & \qw & \qw & \qw & \qw 
	}
\]
}

% ------------------------------------------------------------------
\section{Conclusion}
\label{sec:conclusion}

In this paper we looked at the problem of synthesis of unitary quantum circuits from symbolic expressions as sums-over-paths. We showed that we cannot hope to efficiently synthesize a circuit from a general path sum efficiently, as the problem of checking whether there the path sum represents a unitary transformation is itself \textbf{co-NP}-hard. A stronger result was given recently for the extraction of ZX-diagrams \cite{bkw22}, though their work did not address the complexity of the potentially easier problem of unitarity testing. The problem of unitarity testing for ZX-diagrams is likewise believed to be intractable \cite{w21}.

For the restricted case of Clifford operations, we showed that a circuit can be synthesized efficiently in the form $C_1HC_2$ for Hadamard-free Clifford circuits $C_1$ and $C_2$. For more general path sums we gave a heuristic based on symbolic manipulation and simplification of the sum. We experimentally validated our method, showing that most path sums corresponding to unitary transformations can in fact be synthesized. Moreover, our algorithm is capable of producing natural, high-level circuit designs for some path sums, including the quantum Fourier transform. It remains as a course of future work however to develop a complete synthesis algorithm, as well as to reduce the cost of synthesized circuits.

% --------------------------------------------------------------------
\bibliographystyle{eptcs} 
\bibliography{extraction}

% --------------------------------------------------------------------
\appendix

\section{Finding generalized permutations}\label{app:genperm}

In this appendix we detail our method for finding a generalized permutation in step 2.(a) of \cref{alg:unitary}.

Compared to Clifford operators, simplification via \cref{prop:eqns} may not always leave the path sum in a reducible state. For instance, the path sum expression below, corresponding to a unitary transformation, is fully reduced with respect to \cref{eq:e,eq:i,eq:u}:
\[
	\frac{1}{\sqrt{2^2}}\sum_{y_1,y_2}i^{x_2y_1 - x_2y_2}(-1)^{x_1y_1 + x_1y_2 + x_2y_1y_2}\ket{y_1}\ket{y_2}
\]
However, neither $y_1$ nor $y_2$ are reducible due to the quadratic terms $x_2y_1$ and $x_2y_2$ in the exponent of $i$. At the moment, it is unclear how to proceed symbolically to find a generalized permutation that will make either path variable reducible in the above expression.

Our heuristic method of producing a generalized permutation proceeds in increasingly costly circuit stages in an attempt to synthesize as efficient circuits as possible. Rather than attempt to synthesize a distinct generalized permutation for every path variable $y$ as described in \cref{alg:unitary}, we first apply a sequence of simplification stages generically to both reduce the redundant synthesis work, and produce simpler circuits in practice. The sequence of stages is given in \cref{alg:genperm}, and the individual synthesis steps are described in detail below.

\begin{algorithm}
\caption{Generalized permutation synthesis heuristic}
\label{alg:genperm}
\begin{enumerate}
	\item Apply affine simplifications to the output state $\ket{f(\vec{x},\vec{y})}$
	\item Apply non-linear simplifications to the phase $e^{2\pi iP(\vec{x},\vec{y})}$
	\item Apply non-linear simplifications to the output state $\ket{f'(\vec{x},\vec{y})}$
	\item Apply non-linear simplifications to the phase $e^{2\pi iP'(\vec{x},\vec{y})}$
	\item If no path variable is reducible, attempt degree reduction on each variable
\end{enumerate}
\end{algorithm}

\paragraph{Affine simplifications}

As $X$ and $CNOT$ gates are relatively inexpensive, the first stage of our generalize permutation synthesis attempts to simplify the output $\ket{f(\vec{x},\vec{y})}$ of the path sum as much as possible using only these affine transformations. In order to reduce the number of high-degree terms, which would otherwise require expensive multiply-controlled Toffoli gates, we perform affine simplifications on a \emph{linearization} of $f$. Specifically, we write each $f_i(\vec{x}, \vec{y})$ as a sparse vector $\vec{u_i}\in\Z_2^{2^{n+k}}$ using reverse lexociographic order for the encoding of monomials, then set $A = \begin{bmatrix} u_1 \; u_2 \; \dots \; u_n \end{bmatrix}^T$ and use Gaussian elimination to compute a sequence of $CNOT$ gates reducing $A$ to echelon form. The example below illustrates our method.

\begin{example}
Consider the path sum $\ket{x_1}\ket{x_2}\ket{x_3\oplus x_1x_2}\ket{x_4\oplus x_1x_2}$. This could naturally by synthesized using two non-linear Toffolis to eliminate the $x_1x_2$ terms from the third and forth qubits. The resulting circuit is shown below:
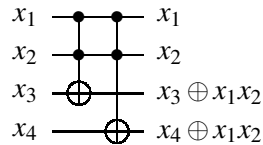
\[
	\Qcircuit @C=.5em @R=.5em @!R {
		\lstick{x_1} & \ctrl{2} & \ctrl{3} & \rstick{x_1} \qw  \\
		\lstick{x_2} & \ctrl{1} & \ctrl{2} & \rstick{x_2} \qw  \\
		\lstick{x_3} & \targ & \qw & \rstick{x_3\oplus x_1x_2} \qw  \\
		\lstick{x_4} & \qw & \targ & \rstick{x_4\oplus x_1x_2} \qw  \\
	}
\]
Alternatively, we can write the output as a (sparse) linear system over all monomials in $x_1,x_2,x_3,x_4$ as shown below:
\[
	\begin{matrix}
		& x_1x_2 & x_4 & x_3 & x_2 & x_1 \\ \hline
		x_1& 0 & 0 & 0 & 0 & 1 \\
		x_2 & 0 & 0 & 0 & 1 & 0 \\
		x_3\oplus x_1x_2 & 1 & 0 & 1 & 0 & 0 \\
		x_4\oplus x_1x_2 & 1 & 1 & 0 & 0 & 0
	\end{matrix}
\]
We use reverse lexicographic order so that reduction to echelon form will prioritize the number of high degree terms. Reducing this to echelon form results in a single $CNOT$ gate and reduces the state to $\ket{x_1}\ket{x_2}\ket{x_3\oplus x_1x_2}\ket{x_4\oplus x_3}$. Synthesizing this remaining transformation gives the overall circuit
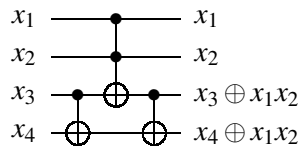
\[
	\Qcircuit @C=.5em @R=.5em @!R {
		\lstick{x_1} & \qw & \ctrl{2} & \qw & \rstick{x_1} \qw  \\
		\lstick{x_2} & \qw & \ctrl{1} & \qw & \rstick{x_2} \qw  \\
		\lstick{x_3} & \ctrl{1} & \targ & \ctrl{1} & \rstick{x_3\oplus x_1x_2} \qw  \\
		\lstick{x_4} & \targ & \qw & \targ & \rstick{x_4\oplus x_1x_2} \qw  \\
	}
\]
\end{example}

\paragraph{Phase simplifications}

To reduce and simplify the number of terms in the phase $e^{2\pi iP(\vec{x},\vec{y})}$ of a path sum controlled $R_Z$ gates with continuous parameters are used. In particular, given an $n$-dimensional path sum $e^{2\pi i \theta\prod_i x_i}\ket{\vec{x}}$, we can reduce the phase term by applying a $\Lambda_n(R_Z(-\theta))$ gate, since
\[
	\Lambda_n(R_Z(-\theta)):e^{2\pi i \theta\prod_i x_i}\ket{\vec{x}}\mapsto \ket{\vec{x}}.
\]
As the output of the path sum is in some state $\ket{f(\vec{x},\vec{y})}$, to apply the above rule we first need to apply a \emph{change of frame} by setting $f_i(\vec{x},\vec{y}) = z_i$ and writing the phase polynomial $P(\vec{x},\vec{y})$ as $P'(\vec{x},\vec{y},\vec{z})$. This is achieved by, for each $f_i$, letting $l$ be the \emph{largest} (non-zero) degree term of $f_i$ and substituting $l\gets z_i \oplus l \oplus f_i(\vec{x},\vec{y})$ in the path sum.

\begin{example}
	Consider the irreducible path sum $\frac{1}{\sqrt{2}}\sum_{y_1}\omega^{-x_1}i^{x_1y_1 - x_2y_1}(-1)^{x_1x_2y_1}\ket{x_1\oplus x_2}\ket{y_1}$. Substituting $[x_2 \gets z_1 \oplus x_1$, $y_1 \gets z_2]$ gives the re-framed path sum
	\[
		\omega^{-x_1}i^{-z_1z_2}(-1)^{x_1z_2}\ket{z_1}\ket{z_2}.
	\]
	Applying a controlled $S$ gate to eliminate the term $i^{-z_1z_2}$ and rolling back the substitutions gives
	\[
		\omega^{-x_1}(-1)^{x_1y_1}\ket{x_1\oplus x_2}\ket{y_1}.
	\]
	The variable $y_1$ is now reducible, so we can finish synthesis by applying a Hadamard to the second qubit, then synthesizing the final generalized permutation $\ket{x_1x_2} \mapsto \omega^{-x_1}\ket{x_1\oplus x_2}\ket{x_1}$. The resulting circuit is given below:
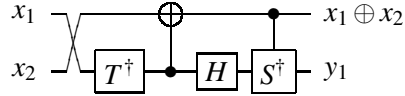
\[
	\Qcircuit @C=.5em @R=.5em @!R {
		\lstick{x_1} & \qw & \link{1}{-1} & \qw & \targ & \qw & \ctrl{1} & \rstick{x_1\oplus x_2} \qw  \\
		\lstick{x_2} & \qw & \link{-1}{-1} & \gate{T^\dagger} & \ctrl{-1} & \gate{H} & \gate{S^\dagger} & \rstick{y_1} \qw  \\
	}
\]
\end{example}

In our implementation, we apply phase simplifications both before and after non-linear simplifications in the state. This is so that we can effectively utilize high degree terms in the state to simplify high degree terms in the phase with phase gates on fewer qubits. The following example illustrates this effect.

\begin{example}
	Consider the path sum $\omega^{x_3 + x_1x_2}i^{-x_1x_2x_3}\ket{x_1}\ket{x_2}\ket{x_3\oplus x_1x_2}$. Eliminating the $x_1x_2$ term in qubit $3$ before simplifying the phase results in the following circuit:
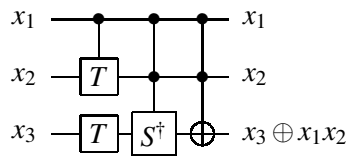
\[
	\Qcircuit @C=.5em @R=.5em @!R {
		\lstick{x_1} & \qw & \ctrl{1} & \ctrl{2} & \ctrl{2}  & \rstick{x_1} \qw  \\
		\lstick{x_2} & \qw & \gate{T} & \ctrl{1} & \ctrl{1}  & \rstick{x_2} \qw  \\
		\lstick{x_3} & \qw & \gate{T} & \gate{S^\dagger} & \targ  & \rstick{x_3\oplus x_1x_2} \qw  \\
	}
\]
However, by re-framing the sum with the substitution $[x_1\gets z_1, x_2\gets z_2, z_1z_2 \gets z_3 \oplus x_3]$ we find
\[
	\omega^{x_3 + x_1x_2}i^{-x_1x_2x_3}\ket{x_1}\ket{x_2}\ket{x_3\oplus x_1x_2} \equiv \omega^{z_3}\ket{z_1}\ket{z_2}\ket{z_3}
\]
which can now be simplified with a single $T$ gate. Note that the substitution is applied left to right, rather than as a simultaneous substitution. The resulting circuit is shown below:
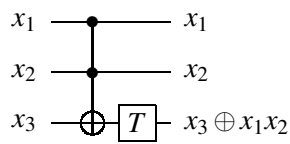
\[
	\Qcircuit @C=.5em @R=.5em @!R {
		\lstick{x_1} & \qw & \ctrl{2} & \qw & \rstick{x_1} \qw  \\
		\lstick{x_2} & \qw & \ctrl{1} & \qw & \rstick{x_2} \qw  \\
		\lstick{x_3} & \qw & \targ & \gate{T} & \rstick{x_3\oplus x_1x_2} \qw  \\
	}
\]
The final substitution $z_1z_2 \gets z_3 \oplus x_3$ may seem counter-intuitive, as we could instead substitute $x_3 \gets z_3 \oplus z_1z_2$. We choose a monomial of maximal degree to substitute in order to avoid inadvertently increasing the degreee of the phase polynomial. For instance, if the initial path sum was instead $\omega^{x_3}\ket{x_1}\ket{x_2}\ket{x_3\oplus x_1x_2}$, substituting $x_3 \gets z_3 \oplus z_1z_2$ results in a re-framed sum of $\omega^{z_3 + z_1z_2}i^{-z_1z_2z_3}\ket{z_1}\ket{z_2}\ket{z_3}$ and the final circuit
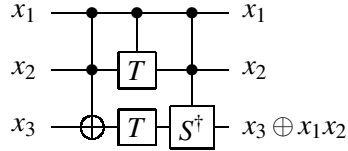
\[
	\Qcircuit @C=.5em @R=.5em @!R {
		\lstick{x_1} & \qw & \ctrl{2} & \ctrl{1} & \ctrl{2} & \rstick{x_1} \qw  \\
		\lstick{x_2} & \qw & \ctrl{1} & \gate{T} & \ctrl{1} & \rstick{x_2} \qw  \\
		\lstick{x_3} & \qw & \targ & \gate{T} & \gate{S^\dagger}  & \rstick{x_3\oplus x_1x_2} \qw  \\
	}
\]
By substituting the highest degree monomial instead, we avoid this issue and synthesize the simpler circuit placing the $T$ gate to the left of the Toffoli.
\end{example}

\paragraph{Non-linear simplifications}

The non-linear simplification step of our synthesis algorithm reduces the number of non-linear terms in the output $\ket{f(\vec{x},\vec{y})}$ by applying multiply-controlled Toffoli gates $\Lambda_k(X)$ via the rule
\[
	\Lambda_k(X):\ket{\vec{x}}\ket{f \oplus \prod_i x_i} \mapsto \ket{\vec{x}}\ket{f}.
\]
Our method for non-linear simplifications uses a na\"{i}ve heuristic whereby a set of variables 
\[
	V= \{v \mid f_i(\vec{x},\vec{y}) = v \text{ for some } i\}
\] 
is computed. Any term in $f(\vec{x},\vec{y})$ which is a product of variables contained in $V$ is then eliminated with an appropriately controlled Toffoli gate.

This method is far from optimal, and in particular misses cases which can be factorized as a cascade of Toffoli gates. While better reversible synthesis methods exist, the lack of a known permutation to synthesize \textit{a priori} in our case makes it difficult to apply such methods directly. An interesting avenue for future work would be to re-synthesize the permutation discovered through this process of simplification using state-of-the-art methods.

\paragraph{Degree reduction}

In many cases, the simplifications previously described fail to leave some path variable in a reducible position. When this happens, our last resort is to fall back to an exponential time procedure we call \emph{degree reduction}. The idea of is to reduce the degree of the (non-Boolean) parts of a phase polynomial restricted to a particular variable, as these terms serve as roadblocks for reduction. This can in some cases be accomplished by applying variable substitutions in such a way as to cancel out terms involving a particular variable.

To illustrate degree reduction, recall the irreducible path sum expression from the beginning of this section,
\[
	\frac{1}{\sqrt{2^2}}\sum_{y_1,y_2}i^{x_2y_1 - x_2y_2}(-1)^{x_1y_1 + x_1y_2 + x_2y_1y_2}\ket{y_1}\ket{y_2}
\]
The exponent of $i$ cannot be directly reduced via simple phase simplifications, as both terms depend on $x_2$. However, we can \emph{indirectly} eliminate one of these terms by applying \cref{eq:v}, substituting $y_1$ with $y_1\oplus y_2$:
\begin{align*}
	\frac{1}{\sqrt{2^2}}&\sum_{y_1,y_2}i^{x_2y_1 - x_2y_2}(-1)^{x_1y_1 + x_1y_2 + x_2y_1y_2}\ket{y_1}\ket{y_2} \\
		&= \frac{1}{\sqrt{2^2}}\sum_{y_1,y_2}i^{x_2(y_1 \oplus y_2) - x_2y_2}(-1)^{x_1(y_1 \oplus y_2) + x_1y_2 + x_2(y_1 \oplus y_2)y_2}\ket{y_1 \oplus y_2}\ket{y_2} && \text{ by \cref{eq:v}} \\
		&= \frac{1}{\sqrt{2^2}}\sum_{y_1,y_2}i^{x_2y_1}(-1)^{x_1y_1 + x_2y_2}\ket{y_1 \oplus y_2}\ket{y_2}
\end{align*}
The final expression above can be simplified to $\frac{1}{\sqrt{2^2}}\sum_{y_1,y_2}i^{x_2y_1}(-1)^{x_1y_1 + x_2y_2}\ket{y_1}\ket{y_2}$ by applying a $CNOT$ gate, which leaves $y_2$ in a reducible position. The resulting circuit is given below:
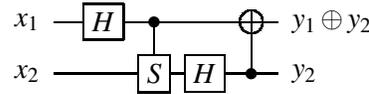
\[
	\Qcircuit @C=.5em @R=.5em @!R {
		\lstick{x_1} & \qw & \gate{H} & \ctrl{1} & \qw & \targ & \rstick{y_1\oplus y_2} \qw  \\
		\lstick{x_2} & \qw & \qw & \gate{S} & \gate{H} & \ctrl{-1} & \rstick{y_2} \qw  \\
	}
\]

The above example is relatively easy to spot, but more complicated cases may require substitution of multiple variables, or even non-linear substitutions. Our heuristic revolves around computing a type of \emph{cover} for the quotient $2(P/y)$, where $y$ is the candidate for degree reduction.

\begin{lemma}\label{lem:poly}
Let $P\in\mathbb{R}[y,x_1,\dots, x_n]$ such that $4(P/y)\equiv 0 \mod 2$. If there exists a set $S\subseteq \{1,\dots,n\}$ such that $4(P/x_i) \equiv 0 \mod 2$ for all $i\in S$ and
\[
	\sum_{i\in S} 2(P/x_i) \equiv 2(P/y) \mod 2
\]
then $2(P[x_i \gets \overline{x_i \oplus y} \mid i\in S]/y) \equiv 0 \mod 2$
\end{lemma}
\begin{proof}
First recall that $\overline{x \oplus y} = x + y - 2xy$. Hence
\[
	2P[x_i \gets \overline{x_i \oplus y} \mid i\in S] = 2P + \sum_{i\in S} 2y(P/x_i) + \sum_{i\in S}4x_iy(P/x_i)
\]
Taking the quotient by $y$ we see
\begin{align*}
	2(P[x_i \gets \overline{x_i \oplus y} \mid i\in S]/y) &= 2(P/y) + \sum_{i\in S} 2(P/x_i) + \sum_{i\in S}4x_i(P/x_i) \\
		&\equiv 2(P/y) + 2(P/y) \mod 2 \\
		&\equiv 0 \mod 2
\end{align*}
\end{proof}

In the context of path sums, \cref{lem:poly} tells us that if $e^{2\pi iP(\vec{x}, \vec{y})}$ can be written as $i^{y_iQ(\vec{x}, \vec{y})}e^{2\pi iR(\vec{x}, \vec{y})}$ for some $i$, and there exists a subset of path variables $\{y_j \mid j \neq i\}$ such that $e^{2\pi iP(\vec{x}, \vec{y})} = i^{y_jQ_j(\vec{x}, \vec{y})}e^{2\pi iR_j(\vec{x}, \vec{y})}$ and $i^{\sum_j Q_j(\vec{x}, \vec{y})} = i^{Q(\vec{x}, \vec{y})}$, then the simultaneous substitution $y_j \gets y_j \oplus y_i$ will eliminate the term $i^{y_iQ(\vec{x}, \vec{y})}$. Additional simplifications in the state may then be further required to leave the path sum in a reducible state, as in the previous above.

\end{document}